%% file: main.tex
\documentclass{article}
\usepackage{graphicx} 
\usepackage{amsmath}
\def\withcolors{1}
\def\withnotes{1}
\usepackage{ccanonne}
\usepackage{multicol}
\usepackage{multirow}
\usepackage{subcaption}
\usepackage{mleftright}
\usepackage{algpseudocode}
\usepackage{physics}
\usepackage[most]{tcolorbox}
\usepackage{etoolbox}

\newcolumntype{C}{>{\centering\arraybackslash}p{0.27\textwidth}}
\newcommand{\overlap}[2]{\langle #1, #2 \rangle}
\newcommand{\target}{\rho}
\newcommand{\lab}{\sigma}
\newcommand{\ftarget}[1]{\hat{\target}(#1)}
\newcommand{\flab}[1]{\hat{\lab}(#1)}
\newcommand{\pauliset}{\mathcal{Q}}
\newcommand{\parity}[2]{\chi_{#1}(#2)}
\newcommand{\Var}[1]{\mathrm{Var}\left(#1\right)}
\newcommand{\fest}[2]{\Tilde{F}(#1, #2)}
\newcommand{\basesset}{\mathcal{B}}
\newcommand{\biasnorm}[2]{\mathcal{N}_{#1}(#2)}
\newcommand{\statespace}{\mathcal{D}(\mathbb{C}^d)}
\newcommand{\knap}[1]{\mathcal{W}_\gamma(#1)}
\newcommand{\randgraph}[2]{\mathcal{E}_{#2}(#1)}
\newcommand{\qcommuteset}{\mathcal{S}}

\newcommand{\stab}[2]{\mathbf{Stab}_{#1}[#2]}
\newcommand{\shadownorm}[1]{\|#1\|_{\mathrm{shadow}}^{\mathrm{(pauli)}}}

\newtheorem{theorem}{Theorem}
\newtheorem{lemma}{Lemma}
\newtheorem{cor}{Corollary}

\newtheorem{observation}{Observation}
\crefname{observation}{Observation}{observations}
\newcommand{\returnifnotempty}[1]{\if\relax\detokenize{#1}\relax\else \textmd{(#1)}\fi}

\newtcolorbox{basicprotocol}[2]{
   colframe=black,       
    colback=white,        
    colbacktitle=white,   
    coltext=black,        
    coltitle=black,       
    boxrule=0.5pt,
  fonttitle=\bfseries,      
  title=#1 \textmd{(#2)},                  
  sharp corners
}

\newtcolorbox{mainprotocol}[2]{
   colframe=black,       
    colback=white,        
    colbacktitle=white,   
    coltext=black,        
    coltitle=black,       
    boxrule=0.5pt,
  fonttitle=\bfseries,      
  title=#1 \returnifnotempty{#2},                  
  sharp corners
}

\newtcolorbox{extraprotocol}[2]{
   colframe=black,       
    colback=white,        
    colbacktitle=white,   
    coltext=black,        
    coltitle=black,       
    boxrule=0.5pt,
  fonttitle=\bfseries,      
  title=#1 \ifstrempty{#2}{
}{
    (#2)
},                  
  sharp corners
}

\input{locdef}

\title{Sublinear Copies Suffice for Fidelity Estimation with Pauli Measurements}
\author{
    \begin{tabular}[t]{C@{\extracolsep{6.5em}} C}
   Jayadev Acharya &Abhilash Dharmavarapu \\
 Cornell University & Cornell University\\ 
\small \texttt{acharya@cornell.edu} &\small \texttt{ad2255@cornell.edu} 
\end{tabular}
\vspace{2ex}\\
\begin{tabular}[t]{C@{\extracolsep{6.5em}} C}
    Yuhan Liu & Nengkun Yu \\
Rice University & Stony Brook University\\ 
\small \texttt{yuhan-liu@rice.edu} &\small \texttt{nengkun.yu@cs.stonybrook.edu} 
\end{tabular}}

\begin{document}

\maketitle

\begin{abstract}

We present the first worst-case, sublinear protocol to estimate the fidelity between a known target state and unknown lab-prepared state using Pauli basis measurements. More precisely, for $d$-dimensional states, we estimate the fidelity to $\pm \varepsilon$ using $o(d^{0.9675})/\varepsilon^2$  copies, which applies to all target states. Furthermore, our results establish improved guarantees for Pauli shadow tomography and Distributed Inner Product Estimation.

\end{abstract}
\input{introduction}

\input{related}

\input{technical_overview}

\input{preliminaries}

\input{paulibasis}
\input{dipe}
\input{additional}

\section{Acknowledgments}
  The authors did not consult AI for the ideas and proofs presented in this paper. AD thanks Saravanan Kandasamy for the insightful discussions. JA, AD and NY are supported by the National Science Foundation under Grant No. CCF-2553759. YL is supported by the Naval Research (ONR) grant N00014-23-1-2737.

\bibliography{refs}
\bibliographystyle{alpha}
\end{document}

%% file: locdef.tex
\newcommand{\unif}{{\mathbf{u}}}

\newcommand{\nqubits}{{N}}

\newcommand{\x}{\mathbf{x}}

\def\multiset#1#2{\ensuremath{\left(\kern-.3em\left(\genfrac{}{}{0pt}{}{#1}{#2}\right)\kern-.3em\right)}}

\newcommand{\POVM}{\mathcal{M}}

%% file: introduction.tex
\section{Introduction}

Recent experiments have seen success in preparing multipartite entangled states in superconducting-qubit processors, trapped-ion systems, and neutral-atom arrays \cite{Cao2023,Moses2023,Shaw2024}. It is becoming increasingly more important to certify that the physical systems are realized as anticipated \cite{PhysRevLett.124.010504}. 
A natural approach to this problem is quantum state tomography, which reconstructs a complete mathematical description of the system. Although powerful, tomography can be quite resource intensive, as it requires many measurements to accurately reconstruct the density matrix of the state. Quantum fidelity therefore serves as a widely used measure of experimental quality, since estimating it for a specified target state is much more attainable than tomography. As a result, there has been a large body of work on fidelity estimation between a known state $\ket{\psi}$ and an unknown state $\sigma$.

The optimal copy-complexity of this task is achieved using the measurement $\{\op{\psi}{\psi}, I-\op{\psi}{\psi}\}$. However, preparing this measurement is often just as difficult as preparing $\op{\psi}$ itself. Fidelity estimation starts to become more interesting when considering quantum measurement settings that are easily realizable in the lab. One of the most experimentally friendly measurement settings is Pauli measurements:  \cite{flamia2011direct,PhysRevLett.107.210404} show that a linear number of samples, $O(d)$, of $\sigma$ suffices for fidelity estimation, i.e. output $\langle \psi|\sigma\ket{\psi}\pm \eps$ using a linear number of measurements. Since then, an intriguing question has arisen:

\begin{center}
    Is $\bigTheta{d}$ the fundamental barrier for fidelity estimation with Pauli measurements?
\end{center}

 It turns out that it isn't! We provide a protocol that estimates the quantum overlap, up to precision $\eps$, between a known state and a sublinear number of copies, $o(d^{0.9675}/\eps^2)$, of an unknown state using Pauli basis measurements. This result holds for \underline{all target states}.

The protocol is built on the wisdom of quantum experimentalists: Pauli observables containing identity do not require separate measurement choices; their expectation values can be obtained by marginalizing the outcomes of full-weight Pauli measurements. Thus, it suffices to use the \(3^\nqubits\) full-weight Pauli basis measurements, rather than measuring all \(4^\nqubits\) Pauli strings separately. These ideas follow from observations made by~\cite{yu2020sampleefficienttomographypauli} and have been recently shown to be useful for breaking the previous barriers of quantum state tomography \cite{ADLY2025Paulinot,grewal2026paulpure}, and certification \cite{Yu2023AlmostTight,acharya2026nearoptimalmixednesstestingpauli}.

\subsection{Problem Setup}
Given $n$ copies of $\lab$ and a state description of $\target$, the goal is to design a measurement scheme $\POVM^n = (\POVM_1, \ldots, \POVM_n)$ to estimate the overlap $\overlap{\target}{\lab}$ (fidelity when $\target$ is pure) with precision $\eps$ and high constant success probability.
\begin{align*}
    \Pr[|f(\POVM^n(\lab^{\otimes n}), \target) - \overlap{\target}{\lab}| \leq \eps] \geq 1-\delta,
\end{align*}
where $\POVM^n(\lab^{\otimes n})$ are the outcomes of measuring each copy of $\lab$ with $(\POVM_1, \ldots, \POVM_n)$ and $f(\POVM^n(\lab^{\otimes n}), \target)$ is the output of the protocol given the state description and measurement outcomes. We restrict our problem to Pauli basis measurements, i.e. $\POVM_i$ is a Pauli basis measurement for all $i \in [n]$.

\subsection{Results}
We highlight the main result of this paper: a fidelity estimation protocol that uses a sublinear number of Pauli basis measurements for any target state.
\begin{theorem}[Sublinear Fidelity Estimation with Pauli Basis Measurements, informal] 
     There exists a protocol that estimates fidelity with
    \begin{align*}
    n =\frac{o(d^{0.9675})}{\eps^2} \log\frac{1}{\delta}
    \end{align*}
    copies of $\lab$ using Pauli basis measurements, where measurements are chosen independently of $\target$.
\end{theorem}

\noindent 
The sublinear protocol establishes an unbiased estimate of fidelity by uniformly sampling Pauli bases and reconstructing the Pauli coefficients, which is equivalent to the classical shadows protocol introduced in~\cite{huang2020predicting}. Huang et al. show that the variance of this protocol depends on the Pauli shadow norm of the target state. Previous works have established bounds that scale linearly with the Frobenius norm of observable~\cite[Proposition 3]{huang2020predicting} and constant scaling for observables with a bounded degree~\cite[Lemma 11]{huang2023predict}. In order to establish the sublinearity of our protocol, we prove the first sublinear scaling of the Pauli shadow norm for all target states.
\begin{theorem}[Sublinear bound on Pauli Shadow Norm, informal] For any mixed state $\target$, 
\begin{align*}
    \left(\shadownorm{\target}\right)^2 = o(d^{0.9675}) \cdot \|\target\|_2^2,
\end{align*}
where $\shadownorm{\target}$ is the shadow norm under randomized Pauli measurements.
\end{theorem}
Note that this result also holds for any p.s.d. observable $O$ since we can scale the density matrix by factor $\Tr[O]$. The high-level idea behind the bound is to extract a randomized graph structure from $\shadownorm{O}$ in terms of the Pauli coefficients of $\target$. Then, we use ideas from noise-stability and hypercontractivity to bound the expected edge-weight of this graph. An immediate consequence of the shadow norm bound is improved Pauli shadow tomography.

\begin{cor}[Improved Shadow Tomography with Pauli Measurements] There exists a protocol that estimates $\Tr(\lab O_1), \Tr(\lab O_2), \ldots, \Tr(\lab O_M)$ for p.s.d observables $O_1, O_2, \ldots, O_M$ with 
\begin{align*}
n =  \frac{o(d^{0.9675})}{\eps^2} \log M \cdot \|O_{max}\|_2^2
\end{align*}
Pauli measurements, where $\|O_{max}\|_2 = \max_{i=1,2,\ldots,M} \|O_i\|_2$
\end{cor}

An immediate consequence of the estimation of fidelity is that we can distinguish between the cases when $\overlap{\target}{\lab} \geq 1-\frac{\eps}{C}$ and $\overlap{\target}{\lab} \leq 1-\eps$ for some constant $C > 1$. This task is known as \emph{robust certification}, and it can be performed with a sublinear number of Pauli basis measurements using our fidelity estimation protocol.
\begin{cor}[Robust Sublinear Certification with Pauli Basis Measurements]
For constant $C > 1$ and pure state $\target$, there exists a protocol that 
\begin{center}
   \begin{itemize}
       \item Outputs \textbf{accept} with probability at least $1-\delta$ if $\overlap{\target}{\lab} \geq 1-\frac{\eps}{C}$.
       \smallskip
       \item Outputs \textbf{reject} with probability at least $1-\delta$ if $\overlap{\target}{\lab} \leq 1- \eps$.
   \end{itemize} 
\end{center}
   \smallskip
with $n= \frac{o(d^{0.9675})}{\eps^2} \log\frac{1}{\delta}$ copies of $\lab$ using Pauli basis measurements, where measurements are chosen independent of $\target$.
\end{cor}
\noindent Therefore, the results provide a worst-case analogue to~\cite{coladangelo2026robust}. Furthermore, our sublinear shadow norm bound can be used to establish sublinear copy-complexity for estimating the inner product $\Tr[\target \lab]$ with $n$ copies of $\lab$ and $\target$.
\begin{cor}[Sublinear Distributed Inner Product Estimation with Pauli Measurements, informal]
Given $n$ copies of two unknown mixed states $\target$ and $\lab$, there exists a protocol that estimates the inner product $\overlap{\target}{\lab}$ with $n = \frac{o(d^{0.9675})}{\eps^2} \log \frac{1}{\delta}$
copies for sufficiently small $\eps$ using Pauli basis measurements.
\end{cor}

%% file: related.tex
\section{Related work}
\paragraph{Tomography.} The goal of Quantum State Tomography (QST) is to reconstruct the density matrix of a state given that we have measurement access to multiple copies of it. In the unrestricted setting, known as entangled measurements, it has been shown that the copy-complexity is $\bigTheta{rd/\eps^2}$, where $r$ is the rank of the density matrix~\cite{HaahHJWY17, ODonnellW16, ODonnellW17, scharnhorst2025}. 

There has been an increasing interest in performing QST in more restricted, yet practical measurement settings due to the complexity of performing large entangled measurements over all the copies. Single-copy measurements are an example of such setting, as the measurements are performed independently on each copy of the state. It has been shown that the copy-complexity of single-copy QST is $\tildeTheta{dr^2/\eps^2}$~\cite{HaahHJWY17, lowe2022lower, guctua2020fast, KRT14}. Furthermore, the results hold even when measurements can be adaptively chosen from the previous measurement outcomes~\cite{nayak2026optimallowrankquantumstate, chen2023does}. Other lines of work also investigate the precise tradeoffs of limited entanglement on the copy-complexity of QST~\cite{Chen0L24memory,nayak2026optimallowrankquantumstate}.

Single-qubit measurements are even simpler than single-copy measurements, as they are performed individually on each qubit. In this setting, particularly with the subclass of Pauli basis measurements, \cite{ADLY2025Paulinot, acharya2025single} has nearly resolved the copy-complexity to be $\tildeTheta{10^\nqubits/\eps^2}$. It has also been shown, rather surprisingly, that Pauli tomography on pure states nearly matches the copy-complexity of entangled and single-copy measurements~\cite{grewal2026paulpure}.

\paragraph{Certification.} Oftentimes, it is not necessary to fully construct the density matrix to verify degradation of a state. It may be sufficient to test whether or not a state is $\eps$-far, in trace distance, from some known target state. This simpler problem is known as Quantum State Certification and is a fundamental problem in quantum property testing. In the entangled measurement setting, it has been shown that the copy-complexity of this task is $\bigTheta{d/\eps^2}$~\cite{BadescuO019, OW15}. In the single-copy setting, it has been shown that the copy-complexity is $\bigTheta{d^{3/2}/\eps^2}$~\cite{BubeckC020, Chen0HL22}, which holds for adaptively chosen measurements. Target-dependent copy-complexity results have also been established under these measurement settings~\cite{ChenLO22instance, odonnell2025instanceoptimalquantumstatecertification}. In addition, it has recently been shown that the copy-complexity of single-qubit/Pauli mixedness testing (certification of the maximally-mixed state) is $\tildeTheta{\sqrt{10}^\nqubits/\eps^2}$\cite{acharya2026nearoptimalmixednesstestingpauli}.

Pure state certification has also been a problem that has garnered significant interest. Under the entangled and single-copy settings, it is clear that the copy-complexity is $\bigTheta{1/\eps^2}$. However, the problem becomes more interesting when considering single-qubit/Pauli basis measurements. Recently, \cite{coladangelo2026robust} have shown that $\bigTheta{1/\eps^2}$ copies suffice to certify $1-2^{-\Omega(\nqubits)}$ fraction of all states with Pauli basis measurements, improving on the existing $\tildeTheta{1/\eps^2}$ result from~\cite{preskill2024certify}.~\cite{ du2026certifying, coladangelo2026power} improved testing robustness of $1-2^{-\Omega(\nqubits)}$ fraction of states with single-qubit measurements. In addition, \cite{gupta2026suffice} has shown that $\tildeTheta{1/\eps^2}$ copies suffice to verify all states with adaptive single-qubit measurements, while demonstrating a $d^{\Omega(1)}$ lower bound for non-adaptive single-qubit protocols. 
\paragraph{Fidelity and Distributed Inner Product Estimation.} Estimating the fidelity, $\overlap{\target}{\lab}$, allows us to precisely describe how far away one state is from the other. As a result, there has been considerable interest in estimating this quantity to efficiently verify quantum devices. Direct Fidelity Estimation (DFE) has been proposed to efficiently estimate the fidelity between a known target state and an unknown lab-prepared state with two-outcome Pauli observable measurements~\cite{flamia2011direct}. The authors have shown that this task can be done with $\bigO{d/\eps^2}$ copies. Subsequent works discover the effects of Pauli grouping on DFE~\cite{barel2026optimizingresourceboundsdirect, barberaRodriguez2025sampling, barberàrodríguez2026directfidelityestimationjoint}. \cite{fawzi2026optimal} provides nearly tight bounds for DFE dependent on the target state. Fidelity estimation has also been considered with Pauli basis measurements, where the authors from~\cite{sun2025efficientfidelity} show that fidelity estimation can be performed for $1-2^{-\Omega(\nqubits)}$ fraction of states with $\tildeTheta{1/\eps^2}$ copies. 

Another related problem is Distributed Inner Product estimation, where the goal is to estimate $\overlap{\target}{\lab}$ using $n$ copies of $\target$ and $\lab$. When it was first introduced in the unrestricted measurement setting, the authors of~\cite{anshu2022distributed} show that $\bigTheta{1/\eps^2 \lor \sqrt{d}/\eps}$ copies are necessary and can be attained with single-copy measurements. \cite{hinsche2025efficient} demonstrates efficient Distributed Inner Product Estimation for states with limited entanglement and magic via a Pauli sampling protocol that is reminiscent of the original DFE paper. \cite{zheng2026distributed, huang2026local} investigate the average-case and worst-case complexity of Distributed Inner Product Estimation under random structured circuits and local randomized measurements.

%% file: technical_overview.tex
\section{Technical Overview}
We present a high-level overview of our sublinear fidelity estimation protocol with Pauli basis measurements. First, we will get acquainted with the idea of target-agnostic fidelity estimation using Pauli observable measurements, where we show that the scheme attains optimal $\bigO{d}$ copy-complexity. Then, we will extend these ideas with Pauli basis measurements and prove the sublinear copy-complexity of the protocol.
\input{warmup}

\input{variance_overview}

%% file: warmup.tex
\subsection{Warm Up with Pauli Observable Measurements}~\label{sec:warm}To understand the higher-level ideas of the sublinear protocol, we present and analyze a warm-up fidelity estimator with Pauli observable measurements. The warm-up fidelity protocol achieves the same copy-complexity as~\cite{flamia2011direct} but has two advantages. Firstly, the procedure is target-agnostic, so it remains consistent across all target states. In addition, we do not require that the target state is pure. This contrasts the approach in ~\cite{flamia2011direct}, where they use importance sampling from the distribution of Pauli coefficients $\{\ftarget{Q}^2\}_{Q \in \mathcal{Q}}$, which remains a valid probability distribution only when $\target$ is pure. Our alternative target-agnostic protocol is described below.

\begin{basicprotocol}{P0: Warm-up Fidelity Protocol}{Pauli Observable Measurements}
 \textbf{Input: } Access to the state description $\{\ftarget{Q}\}_{Q \in \pauliset}$ and measurement access to $\lab^{\otimes n}$. 
 \bigskip
 
 \textbf{Basis Sampling. } Sample $Q_1, Q_2, \ldots, Q_n$ uniformly from $\pauliset$.
\bigskip

\textbf{Measurement postprocessing. } For each $i \in [n]$, measure copy $\lab_i$ with $\POVM_{Q_i}$ and obtain the outcome $X_{Q_i} \in \{-1,1\}$
\bigskip

\textbf{Combining the outcomes. } Return $\fest{\target}{\lab} = \frac{d^{3/2}}{n}\sum_{i=1}^n \ftarget{Q_i} \cdot X_{Q_i}$
\end{basicprotocol}

The procedure is simple. We estimate $\overlap{\target}{\lab} = \sum_{Q \in \pauliset} \ftarget{Q}\flab{Q}$ by uniformly sampling the Pauli observable measurements and scale it by $\frac{1}{\sqrt{d}}$ to obtain an unbiased estimate of the Pauli coefficient $\hat{\lab}$. Then, we compute the inner product of the estimated lab state coefficients and the known coefficients $\hat{\target}$. We will show that this scheme matches the linear rate of~\cite{flamia2011direct}.
\begin{observation}[Copy-complexity of Warm-up Protocol]
    Protocol P0 estimates fidelity with $n = \bigO{\frac{d}{\eps^2} \log \frac{1}{\delta}}$ copies.
\end{observation}
\begin{proof}
We will prove the correctness and soundness of this protocol. First, we show the estimator is unbiased,
\begin{align*}
    \expectDistrOf{}{\fest{\target}{\lab}} &= d^{3/2} \cdot \expectDistrOf{}{\ftarget{Q_i} \cdot X_{Q_i}} \\
    &= d^2 \cdot \expectDistrOf{\lab,\unif[\pauliset]}{\ftarget{Q} \cdot \frac{X_{Q}}{\sqrt{d}}} \\
    &= \sum_{Q \in \pauliset} \ftarget{Q} \flab{Q}.
\end{align*}
Now, it suffices to bound the variance of $\fest{\target}{\lab}$
\begin{align*}
    \Var{\fest{\target}{\lab}} &= \frac{d^3}{n} \Var{\ftarget{Q} \cdot X_Q} \\
    &\leq \frac{d}{n} \sum_{Q \in \mathcal{Q}} \ftarget{Q}^2 \leq \frac{d}{n},
\end{align*}
where we used the fact that the $\ell_2$ norm of the Pauli coefficients is the purity of the state. Setting $n = 3 \cdot \frac{d}{\eps^2} $, we have that $\Var{\fest{\target}{\lab}} \leq 
\frac{\eps^2}{3}$. By Chebyshev's inequality, 
\begin{align*}
\Pr[|\Tilde{F}(\target, \lab) - \overlap{\target}{\lab}| \geq \eps] \leq \frac{1}{3}.
\end{align*}
With Median-of-Means amplification~\cite[Lemma 6.1]{jerrum1985random}, the protocol returns a valid fidelity estimator with $n = \bigO{\frac{d}{\eps^2} \log \frac{1}{\delta}}$ copies.
\end{proof}
The main idea behind the target-agnostic protocol is taking advantage of the p.s.d and purity constraints in the variance of the estimator. It turns out that we can use the exact same idea to establish our sublinear fidelity estimator with Pauli basis measurements!

%% file: variance_overview.tex
\subsection{Variance Analysis for Pauli Basis Measurements}
We will now extend the ideas of the warm-up estimator to construct a protocol for Pauli basis measurements, which is equivalent to the classical shadows protocol in~\cite{huang2020predicting}.
\begin{mainprotocol}{P1: Fidelity Estimation with Pauli Shadows}{Pauli Basis Measurements}
 \textbf{Input: } Access to the state description $\{\ftarget{Q}\}_{Q \in \pauliset}$ and measurement access to $\lab^{\otimes n}$. 
 \bigskip
 
 \textbf{Basis sampling. } Sample $B_1, B_2, \ldots, B_n$ uniformly from $\basesset= \{X,Y,Z\}^{\otimes \nqubits}$.
\bigskip

\textbf{Measurement postprocessing. } For each $i \in [n]$, measure copy $\lab_i$ with $\POVM_{B_i}$ and obtain the outcome $X_{B_i} \in \{-1,1\}^\nqubits$. Then, construct the shadow overlap estimate
$$s_i =  \frac{1}{\sqrt{d}} \sum_{S \subseteq [\nqubits]} \frac{\ftarget{B_i^S} \parity{S}{X_{B_i}}}{3^{\nqubits - |S|}}.$$

\textbf{Combining the outcomes. } Return $\fest{\target}{\lab} = \frac{3^\nqubits}{n}\sum_{i=1}^n s_i$.
\end{mainprotocol}
We use the notation $B^S = \bigotimes_{i=1}^\nqubits \indic{i \in S} B_i + \indic{i \notin S} I$ and $\parity{S}{x} = \prod_{i \in S} x_i$. We notice that the protocol is similar to the warm-up but with additional post processing steps. Since each Pauli basis learns about $d$ different Pauli coefficients (captured by parity functions $\parity{S}{x}$ on the outcome $x \in \{-1,1\}^\nqubits$), there exists overlap in the learned coefficients across the measurement bases. The $\frac{1}{3^{\nqubits - |S|}}$ factor is introduced to normalize these overlaps and ensure that $\fest{\target}{\lab}$ remains an unbiased estimator of the fidelity. Bounding the variance of these shadow overlaps is the main challenge in proving the sublinear rate.
\paragraph{Reducing to the edge-weight of Random Pauli Graphs. }\cite{huang2020predicting}~showed that the variance of the protocol scales with the Pauli shadow norm,

\begin{align*}
    \Var{\Tilde{F}(\target, \lab)} \leq \frac{\left(\shadownorm{\target}\right)^2}{n},
\end{align*}
where the Pauli shadow norm is defined as the following,
\begin{align*}
    \left(\shadownorm{\target}\right)^2 = \frac{1}{d} \max_{\lab } \expectDistrOf{B \sim \basesset}{ \expectDistrOf{\lab}{\left(\sum_{S \subseteq [\nqubits]} \ftarget{B^S} 3^{|S|} \parity{S}{x} \right)^2}}
\end{align*}
\cite{huang2020predicting} established $\left(\shadownorm{\target} \right)^2 \leq d \Tr[\target^2]$ with a simple Cauchy-Schwarz argument. We provide a strictly stronger bound by extracting the following quantity with H\"older's inequality,
\begin{align*}
  \randgraph{\target}{\gamma} = \sum_{(Q, Q') \in \qcommuteset} \frac{|\ftarget{Q}|^{\frac{2}{1+\gamma}} \cdot |\ftarget{Q'}|^{\frac{2}{1+\gamma}}}{3^{\nqubits - w(Q \cap Q')}},
\end{align*}
where $\qcommuteset$ denotes pairs of Pauli observables that agree on their non-identity entries (known as the set of \emph{single-qubit commuting} observables) and $w(Q \cap Q')$ is the number of intersecting non-identity entries between $Q$ and $Q'$. Notice that $\randgraph{\target}{1}$ is the 2-norm of the target state under a channel that is neither completely-positive nor trace-preserving. 

We will first consider the case where $\gamma$ is $0$. When $\target$ is pure, i.e., $\sum_{Q \in \pauliset} \ftarget{Q}^2 = 1$, we can interpret this quantity as the expected edge weight of the graph on Pauli strings, where the pairs of vertices are sampled according to $Q,Q' \sim \nu =\hat{\target}^2$ with edge weight $3^{-(\nqubits-w(Q \cap Q'))}$. A natural approach is to bound the adjacency matrix by its maximum eigenvalue. We end up with the following,
\begin{align*}
    \randgraph{\target}{0} \leq \left(\frac{4}{3}\right)^\nqubits \sum_{Q \in \pauliset} \ftarget{Q}^4.
\end{align*}
Since the magnitudes of the coefficients are bounded by $\frac{1}{\sqrt{d}}$, we get a resulting bound of $\left(\frac{2}{3}\right)^\nqubits$. Multiplying by the normalization factor $3^\nqubits$, we see that the variance only achieves a linear scaling. However, it can be seen that this bound is not tight since the states attaining the maximum eigenvalue have substantially lower purity. The challenge is therefore to derive a sharper universal bound on the expected edge weight of the Pauli graph.

The key idea is to re-interpret the problem under the context of noise-stability. We factor $\left(\frac{4}{3}\right)^\nqubits$ from the adjacency matrix of the graph and notice that the matrix corresponds to the transition probabilities under a noisy channel. 
\begin{align*}
    \randgraph{\target}{0} = \left(\frac{4}{3}\right)^\nqubits \overlap{\nu}{\mathcal{C} \nu} = \left(\frac{4}{3}\right)^\nqubits \stab{\mathcal{C}}{\nu}, \quad \mathcal{C} = \begin{bmatrix}
            \frac{3}{4} & 0 & 0 & \frac{1}{4} \\
            0 & \frac{3}{4} & 0 & \frac{1}{4} \\
            0 & 0 & \frac{3}{4} & \frac{1}{4} \\
            \frac{1}{4} & \frac{1}{4} & \frac{1}{4} & \frac{1}{4} 
        \end{bmatrix}^{\otimes \nqubits}
\end{align*}
$\mathcal{C}$ is the noisy channel that erases non $I$ characters to $I$ with probability $1/4$ and uniformly re-assigns $I$ to every character. It is not hard to see that the distribution becomes closer to uniform over repeated channel uses. Furthermore, small increases in $p$ drastically reduce the $\ell_p$ norm when the distribution is closer to uniform. The standard approach to solve these noise-stability problems is to use ideas in hypercontractivity, which quantifies contraction of $\ell_p$ norms under noisy channels~\cite{benjamini1999noise}. With a reduction to the hypercontractivity of Boolean functions, we establish the key result:
\begin{align*}
    \stab{\mathcal{C}}{\nu} \leq 4^{\left(1-\frac{2}{1+\lambda}\right) \nqubits} \cdot \|\nu\|_{1+\lambda}^2,
\end{align*}
where $\lambda = \frac{1 + \sqrt{33}}{8}$. The contraction factor $4^{\left(1-\frac{2}{1+\lambda}\right) \nqubits}$ enables us to establish our sublinear rate and extends to $\randgraph{\target}{\gamma}$ for all $\gamma \in [0,1]$. Another quantity that appears in the second moment is the covariance from estimating the coefficients of $\lab$, which we bound with previously well known results on the capacity of depolarizing channels~\cite{king2003capacity}.

We can now bound the variance, which depends on $\randgraph{\target}{\gamma}$ and the lab-state covariance. The extra $\gamma$ factor in $\randgraph{\target}{\gamma}$ comes from splitting the random graph and the lab state covariance with H\"older's inequality. With $\gamma = \lambda$, we minimize the geometric mean between the lab-state covariance and the hypercontractive bound on the edge-weight of the Pauli graph to obtain the sublinear scaling on the Pauli shadow norm,
\begin{align*}
     \left(\shadownorm{\target}\right)^2 \leq 2^{\lambda \nqubits}\sqrt{3}^{(1-\lambda) \nqubits} \|\target\|_2^2
\end{align*}
Thus, we have
\begin{align*}
\Var{\Tilde{F}(\target, \lab)} \leq \frac{2^{\lambda \nqubits}\sqrt{3}^{(1-\lambda) \nqubits}}{n}.
\end{align*}
The copy-complexity is bounded by the geometric interpolation between the linear rate $d$ and the sublinear rate $\sqrt{3}^\nqubits$, and the bias of the interpolation is exactly the hypercontractivity constant! In~\cref{sec:dipe}, we show that these bounds also extend to Distributed Inner Product Estimation with classical Pauli shadows.

%% file: preliminaries.tex
\section{Preliminaries}
\subsection{Quantum mechanics}
An $\nqubits$-qubit pure state $\ket{\phi}$ is represented as a complex unit vector in a $d$-dimensional Hilbert space, where $d = 2^\nqubits$. A mixed state is $\target$ represented as a probabilistic ensemble of pure states $\sum_{i=1}^r \lambda_i \ket{\phi_i}\bra{\phi_i}$ living in $\mathbb{C}^{d \times d}$. We refer to this matrix as the \emph{density matrix}. Therefore, every quantum state can be represented as a hermitian, positive semi-definite linear operator with unit trace.
\begin{align*}
   \statespace \eqdef \left\{\target \in Herm\left(\mathbb{C}^{d \times d}\right) \mid \Tr[\target] = 1 \land \target \succeq 0\right\}.
\end{align*}
Consider when $\nqubits = 1$, then the following matrices form an orthogonal basis for $\statespace$ under the inner product $\overlap{\target}{\lab} = \Tr[\target \lab]$,
\begin{align*}
    I=\begin{bmatrix}1 &0\\0&1\end{bmatrix}, X=\begin{bmatrix}0 &1\\1&0\end{bmatrix}, Z=\begin{bmatrix}1 &0\\0&-1\end{bmatrix}, Y=\begin{bmatrix}0 &-i\\i&0\end{bmatrix}.
\end{align*}
These matrices are known as the \emph{Pauli matrices}, introduced by Wolfgang Pauli in 1927~\cite{pauli1927quantenmechanik}. Every single-qubit state can be represented in terms of the Pauli matrices.
\begin{align*}
    \target = \frac{ \left(\ftarget{I} \cdot I+\ftarget{X} \cdot X + \ftarget{Y} \cdot Y + \ftarget{Z} \cdot Z \right)}{\sqrt{2}} ,
\end{align*}
with necessary constraints $\ftarget{I} = \frac{1}{\sqrt{2}}$, $\|\hat{\target}\|_\infty \leq \frac{1}{\sqrt{2}}$ and $\|\hat{\target}\|_2 \leq 1$. When $\nqubits > 1$, we can extend the basis by tensor product,
\begin{align*}
    \mathcal{D}\left(\bigotimes_{i=1}^k \mathbb{C}^{d_i}\right) = \mathcal{D}\left(\mathbb{C}^{\Pi_{i=1}^k d_i}\right),
\end{align*}
which represents the composition of multiple quantum systems. This results in $\statespace$ being described by the basis formed by the tensor product of Pauli matrices.
\begin{align*}
    \target = \frac{1}{\sqrt{d}} \sum_{Q \in \pauliset} \ftarget{Q} \cdot Q,
\end{align*}
where $\pauliset \eqdef \{X,Y,Z,I\}^{\otimes \nqubits}$ and $\ftarget{I} = \frac{1}{\sqrt{d}}$. We will refer to $\{\ftarget{Q}\}_{Q \in \pauliset}$ as the Pauli coefficients of $\target$. Similar constraints apply for the Pauli coefficients as the $\nqubits = 1$ case.
\begin{observation}[Pauli Coefficient Constraints]\label{obs:pauli-constraints}
    Given an $\nqubits$-qubit state $\target$ with Pauli coefficients $\{\ftarget{Q}\}_{Q \in \pauliset}$, the following holds true
    \begin{align*}
       \|\hat{\rho}\|_\infty \leq \frac{1}{\sqrt{d}}, \quad \|\hat{\rho}\|_2 \leq 1.
    \end{align*}
\end{observation}

An interesting consequence of the Pauli basis is that much of the geometry of the state space can be described in terms of the Pauli coefficients. In particular, the inner product and norm are represented as
\begin{align*}
    \overlap{\target}{\lab} = \Tr[\target \lab] = \sum_{Q \in \pauliset} \ftarget{Q} \cdot \flab{Q}, \; \|\target\|_2^2 = \sum_{Q \in \pauliset} \ftarget{Q}^2 = \|\hat{\target}\|_2^2.
\end{align*}
$\Tr[\target^2]$, known as the \emph{purity} of a quantum state, is especially relevant in quantum information, as it serves to gauge how close $\target$ is to being a pure state. The quantity is maximized at $1$ when $\target$ is pure. 
\subsection{Quantum measurement}
\paragraph{Positive Operator-Valued Measure.}\label{sec:measure} A measurement of a quantum state is represented as a positive operator-valued measure (POVM). A POVM consists of a countable set of operators $\{M_x\}_{x \in \mathcal{X}}$ such that it satisfies the properties
\begin{align*}
    \sum_{x \in \mathcal{X}} M_x = I,\; \forall_{x \in \mathcal{X}} M_x \succeq 0.
\end{align*}
 By \emph{Born's Rule}, the outcome of a POVM follows the distribution
\begin{align*}
\Pr[X = x] = \Tr[\target M_x]
\end{align*}
for $x \in \mathcal{X}$. The POVM constraints ensure normalization and non-negativity of the probability distribution.
\paragraph{Pauli Measurements.} The measurements of interest throughout this paper are Pauli measurements. One of the simplest measurements is the \emph{Pauli observable measurement}, described by the two-outcome POVM:
\begin{align*}
    M^{Q}_{1} = \frac{I + Q}{2},\; M^{Q}_{-1} = \frac{I - Q}{2},
\end{align*}
where $Q \in \pauliset$. If we apply Born's rule, we get following outcome distribution,
\begin{align*}
    X_Q = \begin{cases}
        1 & \text{w.p } \frac{1}{2} + \frac{\sqrt{d} \ftarget{Q}}{2}\\
        -1 & \text{w.p }\frac{1}{2} - \frac{\sqrt{d} \ftarget{Q}}{2}
    \end{cases},
\end{align*}
Thus, $\expectDistrOf{\target}{X_Q} = \sqrt{d} \ftarget{Q}$. So, performing a two-outcome Pauli observable measurement $\POVM_Q$ on $\target$ results in learning about $\ftarget{Q}$.

An extension of Pauli observable measurements is \emph{Pauli basis measurements}. Instead of observing 2-outcomes, the Pauli basis measurement observes $d=2^\nqubits$ outcomes by independently measuring each qubit in the Pauli basis. The POVM is described below,
\begin{align*}
    M^B_{x} = \bigotimes_{i=1}^\nqubits \frac{I + x_i B_i}{2}, \; x \in \{-1,1\}^\nqubits, B \in \basesset,
\end{align*}
where $\basesset = \{X,Y,Z\}^{\otimes \nqubits}$ and $B = \bigotimes_{i=1}^\nqubits B_i$. When we apply Born's rule, we obtain the following,
\begin{align*}
    \Pr[X_B = x] = \frac{1}{\sqrt{d}} \sum_{S \subseteq [\nqubits]} \ftarget{B^S} \parity{S}{x},
\end{align*}
where $B^S = \bigotimes_{i=1}^\nqubits \indic{i \in S} B_i + \indic{i \notin S} I$ and $\parity{S}{x} = \prod_{i \in S} x_i$. We refer to $\parity{S}{x}$ as the $S$-\emph{parity function} on $x$, which has the property: $\chi_{S}(x) \chi_{S'}(x) = \chi_{S \Delta S'}(x)$. We say $Q \triangleright B$ if all of $Q$'s non-identity tensor entries match those of $B$. An important fact is that $\expectDistrOf{\target}{\parity{S}{x}} = \sqrt{d} \ftarget{B^S}$, which means that a single measurement in the basis $B$ learns about $d = 2^\nqubits$ Pauli coefficients.

\subsection{Pauli notations}
We introduce some notation for operations/quantities associated with Pauli observables. First, we define the \emph{Pauli weight}.
\begin{align*}
    w(Q) = \sum_{i = 1}^\nqubits \indic{Q_i \neq I}.
\end{align*}
Notice that for a Pauli observable $Q$, there exists $3^{\nqubits - w(Q)}$ Pauli basis measurements that can learn about $\ftarget{Q}$. We also provide notation for set-theoretic operations one can do between pairs of Pauli observables that satisfy $Q,Q' \triangleright P$ for some $B \in \basesset$,
\begin{align*}
  Q \Delta Q' &\eqdef \bigotimes_{i=1}^\nqubits \indic{Q_i = Q'_i} I + \indic{Q_i \neq Q'_i = I} Q_i + \indic{Q'_i \neq Q_i = I} Q'_i, \\
Q \cup Q' &\eqdef \bigotimes_{i=1}^\nqubits \indic{Q_i = Q'_i} Q_i + \indic{Q_i \neq Q'_i = I} Q_i + \indic{Q'_i \neq Q_i = I} Q'_i, \\
Q \cap Q' &\eqdef \bigotimes_{i=1}^\nqubits \indic{Q_i = Q'_i} Q_i + \indic{Q_i \neq Q'_i = I} I + \indic{Q'_i \neq Q_i = I} I,
\end{align*}
where $Q = \bigotimes_{i=1}^{\nqubits} Q_i$ and $Q' = \bigotimes_{i=1}^{\nqubits} Q'_i$. Note, we do not consider the case $I \neq Q_i \neq Q'_i \neq I$ since the two observables cannot belong to the same measurement group if there is a disagreement on the non-identity factors. For convenience, we denote $\qcommuteset$ as the set of $(Q,Q')$ that share at least one measurement group. With these notations, we provide a useful identity that is analogous to set inclusion-exclusion.
\begin{observation}[Inclusion-Exclusion Principle with Pauli Observables]\label{obs:inc-exc}
For $(Q,Q') \in \qcommuteset$, we have the following inclusion-exclusion identity:
\begin{align*}
    w(Q)+w(Q')-w(Q \cup Q') &= \sum_{i=1}^{\nqubits} \indic{Q_i \neq I} + \indic{Q'_i \neq I} - \indic{Q_i \neq I \lor Q'_i \neq I}\\
                           &=  \sum_{i=1}^{\nqubits} \indic{Q_i = Q'_i \neq I} = w(Q \cap Q').
\end{align*}
\end{observation}
\subsection{Boolean analysis and hypercontractivity} \label{sec:boolean}
We will now discuss some fundamentals of Boolean analysis and hypercontractivity. A more comprehensive review of the topic can be found in~\cite{odonnell2021analysisbooleanfunctions}. A Boolean function is a real-valued function on the Boolean hypercube $\{-1,1\}^\nqubits$. Every Boolean function can be represented as a multi-linear polynomial,
\begin{align*}
    f(x) = \sum_{S \subseteq [\nqubits]} \hat{f}(S) \parity{S}{x},
\end{align*}
where $\hat{f}(S)$ are the Fourier coefficients of $f$. The parity functions $\{\parity{S}{x}\}_{S \subseteq [\nqubits]}$, defined in~\cref{sec:measure}, form an orthogonal basis for Boolean functions under the inner product $\overlap{f}{g} = \sum_{x \in \{-1,1\}^\nqubits}f(x) g(x)$. Under this basis, we have Fourier-analytic identities with inner product and convolution. By Parseval's identity and Convolution Theorem,

\begin{align*}
    \overlap{f}{g} = 2^\nqubits \sum_{S \subseteq [\nqubits]} \hat{f}(S) \cdot \hat{g}(S), \quad f * g(x) = 2^\nqubits \sum_{S \subseteq [\nqubits]} \left(\hat{f}(S) \cdot \hat{g}(S)\right) \parity{S}{x},
\end{align*}
where convolution is defined by $f*g(x) = \sum_{y \in \{-1,1\}^\nqubits}{f(x \oplus y) g(y)}$. The operation $\oplus$ represents bit-wise multiplication. 

\paragraph{Hypercontractivity.} Hypercontractivity is one of the most important results in Boolean analysis, which states that $\ell_q$ norms contract $\ell_{p < q}$ norms significantly when a noisy-operator is acting on the function. The noisy operator of interest is the Bonami-Beckner operator, which flips bits independently with probability $\frac{1}{2} - \frac{1}{2} \lambda$ for $\lambda \in [-1,1]$ and can be represented by the transition probability matrix,
\renewcommand{\arraystretch}{1.25}
\begin{align*}
    T_\lambda = \begin{bmatrix}
        \frac{1}{2} + \frac{1}{2} \lambda & \frac{1}{2} - \frac{1}{2} \lambda \\
        \frac{1}{2} - \frac{1}{2} \lambda & \frac{1}{2} + \frac{1}{2} \lambda
    \end{bmatrix}^{\otimes \nqubits}.
\end{align*}
The operator is also equivalent to
\begin{align*}
    T_\lambda f(x) = \expectDistrOf{y \sim Rad(\lambda)^{\otimes \nqubits}}{f(x \oplus y)} = 2^\nqubits \cdot N_\lambda * f (x),
\end{align*}
where $Rad(\lambda)$ is a Rademacher random variable that takes value $1$ w.p $\frac{1}{2} + \frac{1}{2} \lambda$ and value $-1$ w.p $\frac{1}{2} - \frac{1}{2}\lambda$. $N_\lambda$ is the probability mass function of $Rad(\lambda)^{\otimes \nqubits}$. Evaluating the Fourier expansion of $N_\lambda$ yields the following,
\begin{align*}
    \hat{N_\lambda}(S) = \frac{\overlap{N_\lambda}{\chi_{S}}}{2^\nqubits} &= \frac{1}{2^\nqubits} \sum_{x \in \{-1,1\}^\nqubits} N_\lambda(x) \parity{S}{x} \\
                                   &= \frac{1}{2^\nqubits} \prod_{i \in S} \expectDistrOf{x_i \sim Rad(\lambda)}{x_i} = \frac{\lambda^{|S|}}{2^\nqubits} 
\end{align*}
By Convolution Theorem, we can see that the Bonami-Beckner operator dampens the higher-degree Fourier coefficients.
\begin{align*}
    T_\lambda f(x) = \sum_{S \subseteq [\nqubits]} \lambda^{|S|} \hat{f}(S) \parity{S}{x}.
\end{align*}
 $\stab{T}{f} = \overlap{f}{Tf}$ is known as the \emph{noise-stability} of $f$ under $T$. A useful property that arises from spectral result is the Markov semi-group property: $T_{\lambda_1 \lambda_2} = T_{\lambda_1} T_{\lambda_2}$.  The operator dampens the higher-degree Fourier coefficients of $f$ to make the function look more uniform, resulting in the larger separation between $\ell_p$ norms. This separation is quantified by the Bonami-Beckner Hypercontractivity Theorem.
\begin{lemma}[{\cite[pg. 284]{odonnell2021analysisbooleanfunctions}}, Hypercontractivity Theorem]\label{lem:hypercontractivity}
    Let  $T_\lambda$ be Bonami-Beckner operator over $\nqubits$ bits. For all $f: \{-1,1\}^\nqubits \rightarrow \mathbb{R}$ and $\lambda \leq \sqrt{\frac{p-1}{q-1}}$ for $1 < p \leq q$,
        \begin{align*}
            \|T_\lambda f\|_q \leq 2^{\left(\frac{1}{q} - \frac{1}{p}\right)\nqubits} \cdot \|f\|_p,
        \end{align*}
        where $\|g\|_p = \left(\sum_{x \in \{-1,1\}^\nqubits} |g(x)|^p\right)^{1/p}$.
\end{lemma}
\noindent For this paper, we consider the case when $q=2$, which results in the 2-norm contracting the $1+\lambda^2$ norm.


%% file: paulibasis.tex
\section{The Sublinear Fidelity Protocol}
We now discuss the Pauli basis measurement protocol that estimates fidelity with a sublinear number of copies. The first subsection will analyze the variance of the protocol and introduce a new quantity called the \emph{$\lab$-biased Pauli norm}. The remaining subsections will focus on proving a bound on the biased Pauli norm to establish our sublinear copy-complexity.
\subsection{Variance analysis}\label{sec:variance-analysis}
We now present our sublinear fidelity estimator. In the warm-up protocol, we were able to take advantage of the p.s.d and purity constraints of the target state to bound the variance of the target-agnostic protocol. We will show that, with some additional machinery, the same underlying idea will apply for Pauli basis measurements.

\addtocounter{theorem}{-2}
\begin{theorem}[Sublinear Fidelity Estimation with Pauli Basis Measurements] \label{thm:sublinear-dfe}
    Protocol P1 estimates fidelity with 
    \begin{align*}
    n &= \bigO{\frac{2^{\lambda \nqubits} \sqrt{3}^{(1-\lambda) \nqubits}}{\eps^2} \log \frac{1}{\delta}}
    \end{align*}
    copies, where $\lambda  = \frac{1 + \sqrt{33}}{8}$.
\end{theorem}
\begin{proof}
    We will proceed to prove the correctness and soundness in similar vein to~\cref{sec:warm}. We first notice the protocol produces an unbiased estimate of $\overlap{\target}{\lab}$.
    \begin{align*}
        \expectDistrOf{}{\fest{\target}{\lab}} &= 3^\nqubits \cdot \expectDistrOf{\lab, \unif[\basesset]}{\frac{1}{\sqrt{d}} \sum_{S \subseteq [\nqubits]} \frac{\ftarget{B^S} \parity{S}{X_{B}}}{3^{\nqubits - |S|}}} \\
        &=3^\nqubits \cdot \expectDistrOf{\unif[\basesset]}{\sum_{S \subseteq [\nqubits]} \frac{\ftarget{B^S} \flab{B^S}}{3^{\nqubits - |S|}}}  = \overlap{\target}{\lab},
    \end{align*}
    where the last equality uses the fact that every Pauli observable with $k$ non-identity elements appears in $3^{\nqubits - k}$ measurement groups. We proceed to bound the variance.
    \begin{align*}
        \Var{\fest{\target}{\lab}} &= \frac{9^\nqubits}{n} \Var{\frac{1}{\sqrt{d}} \sum_{S \subseteq [\nqubits]} \frac{\ftarget{B^S} \parity{S}{X_{B}}}{3^{\nqubits - |S|}}} \\
        &\leq \frac{3^\nqubits}{n} \sum_{B \in \mathcal{B}} \expectDistrOf{\lab}{\left(\frac{1}{\sqrt{d}} \sum_{S \subseteq [\nqubits]} \frac{\ftarget{B^S} \parity{S}{X_{B}}}{3^{\nqubits - |S|}}\right)^2}
    \end{align*}
    Most of the analysis for the sublinear protocol is focused on the following expression:
    \begin{align}
        \biasnorm{\lab}{\target} \eqdef  \sum_{B \in \mathcal{B}} \expectDistrOf{\lab}{\left(\frac{1}{\sqrt{d}} \sum_{S \subseteq [\nqubits]} \frac{\ftarget{B^S} \parity{S}{X_{B}}}{3^{\nqubits - |S|}}\right)^2}, \label{def:mixed-norm}
    \end{align}
    which we refer to as the \emph{$\lab$-biased Pauli norm} of $\target$. A simple application of Cauchy Schwarz leads to the following,
   \begin{align*}
       \biasnorm{\lab}{\target} &\leq \frac{1}{d} \sum_{B \in \basesset} \left(\sum_{S \subseteq [\nqubits]} \frac{\ftarget{B^S}^2}{3^{\nqubits - |S|}}\right) \cdot \left(\sum_{j=0}^\nqubits \frac{1}{3^{\nqubits - j}} \binom{N}{j} \right) \\
       &= \frac{d}{3^\nqubits} \sum_{B \in \basesset} \sum_{S \subseteq [\nqubits]} \frac{\ftarget{B^S}^2}{3^{\nqubits - |S|}} = \frac{d}{3^\nqubits} \sum_{Q \in \pauliset} \ftarget{Q}^2 \leq \frac{d}{3^\nqubits}.
   \end{align*} 
   Substituting this bound, it can be seen that $n = \bigO{d/\eps^2}$ copies is sufficient for Protocol P1 to succeed with high probability. The main question is: 
   \begin{center}
   Can we show $\biasnorm{\lab}{\target} \leq \frac{o(d)}{3^\nqubits}$ with more refined analysis?
   \end{center}
   Given deeper understanding of the underlying structures present in $\biasnorm{\lab}{\target}$, we can indeed give a sublinear bound on the biased Pauli norm.
   \begin{theorem}[Sublinear Bound on Pauli Shadow Norm]~\label{thm:bias-norm} Let $\lambda = \frac{1+\sqrt{33}}{8}$, then
   \begin{align*}
       \frac{1}{3^\nqubits} \left(\shadownorm{\target}\right)^2 = \max_{\lab \in \statespace} \biasnorm{\lab}{\target} \leq  \frac{2^{\lambda \nqubits} \sqrt{3}^{(1-\lambda)\nqubits}}{3^\nqubits} \|\target\|_2^2.
   \end{align*}
   \end{theorem}
   \noindent This key result is proven in the later sections. With this bound, we can see that
   \begin{align*}
       \Var{\fest{\target}{\lab}} \leq \frac{2^{\lambda \nqubits} \sqrt{3}^{(1-\lambda)\nqubits}}{n}.
   \end{align*}
   Setting $n = 3 \cdot \frac{2^{\lambda \nqubits} \sqrt{3}^{(1-\lambda) \nqubits}}{\eps^2}$ and applying Chebyshev's inequality establishes a constant success probability protocol. Then, we apply Median-of-Means amplification to have $n = \bigO{\frac{2^{\lambda \nqubits} \sqrt{3}^{(1-\lambda) \nqubits}}{\eps^2} \log \frac{1}{\delta}}$.
\end{proof}
\subsection{Reducing biased Pauli norm to random Pauli graphs} \label{sec:biased-norm}
We have shown that the performance of the protocol scales with $\biasnorm{\lab}{\target}$. In this section, we will dissect the quantity and show that it can be bounded by the edge weight of random Pauli graphs. We state the intermediary result below.
\begin{lemma}[Reducing $\lab$-Biased Pauli Norm to Random Pauli Graphs]~\label{lem:reduction} For any $0 \leq \gamma < 1$,     
\begin{align*}
    \biasnorm{\lab}{\target} \leq \sqrt{\frac{\randgraph{\target}{\gamma}^{1+\gamma} \cdot \knap{\lab}^{1-\gamma}}{d}},
\end{align*} 
where $\randgraph{\target}{\gamma}$ and $\knap{\lab}$ are defined as
\begin{align*}
    \randgraph{\target}{\gamma} \eqdef \sum_{(Q, Q') \in \qcommuteset} \frac{|\ftarget{Q}|^{\frac{2}{1+\gamma}} \cdot |\ftarget{Q'}|^{\frac{2}{1+\gamma}}}{3^{\nqubits - w(Q \cap Q')}}, \; \knap{\lab} \eqdef \sum_{(Q, Q') \in \qcommuteset} \frac{|\flab{Q \Delta Q'}|^{\frac{2}{1-\gamma}}}{3^{\nqubits - w(Q \cap Q')}}.
\end{align*}
\end{lemma}
Notice, $\randgraph{\target}{0}$ is the expected edge weight of a random graph where the vertices $Q \in \pauliset$ are independently sampled from the distribution $\{\ftarget{Q}^2\}_{Q \in \pauliset}$. We proceed with the proof of this graph reduction.
\begin{proof}
We will expand $\biasnorm{\target}{\lab}$,
\begin{align*}
    \biasnorm{\target}{\lab} &=  \sum_{B \in \basesset} \expectDistrOf{\sigma}{\frac{1}{\sqrt{d}} \sum_{S,S' \subseteq [\nqubits]} \frac{\hat{\target}(B^{S}) \cdot \hat{\target}(B^{S'}) \cdot \chi_{S \Delta S'}(X_B)}{\sqrt{d} \cdot 3^{\nqubits - |S|} \cdot 3^{\nqubits - |S'|}}} \\
    &= \frac{1}{\sqrt{d}} \sum_{B \in \basesset} \sum_{Q,Q' \triangleright B} \frac{\hat{\target}(Q) \cdot \hat{\target}(Q') \cdot \hat{\lab}(Q \Delta Q')}{3^{\nqubits - w(Q)} \cdot 3^{\nqubits - w(Q')}}. 
\end{align*}
We notice that there are exactly $3^{\nqubits - w(Q \cup Q')}$ bases that cover both $Q$ and $Q'$. Thus,
\begin{align*}
    \biasnorm{\target}{\lab} &= \frac{1}{\sqrt{d}}\sum_{(Q,Q') \in \qcommuteset}  \frac{\hat{\target}(Q) \cdot \hat{\target}(Q') \cdot \hat{\lab}(Q \Delta Q')}{3^{\nqubits - w(Q)} \cdot 3^{\nqubits - w(Q')}} \cdot 3^{\nqubits - w(Q \cup Q')} \\
    &= \frac{1}{\sqrt{d}} \sum_{(Q,Q') \in \qcommuteset}\frac{\hat{\target}(Q) \cdot \hat{\target}(Q') \cdot \hat{\lab}(Q \Delta Q')}{3^{\nqubits - w(Q \cap Q')}},
\end{align*}
where the last line used~\cref{obs:inc-exc}. By Holder's inequality on functions over $\qcommuteset$,
\begin{align*}
 \biasnorm{\target}{\lab} \leq \sqrt{\frac{\randgraph{\target}{\gamma}^{1+\gamma}\cdot \knap{\lab}^{1-\gamma}}{d}}.
\end{align*}
\end{proof}
\noindent In the next two sections, we will bound $\randgraph{\target}{\gamma}$ and $\knap{\lab}$, respectively.
\subsection{Edge weight of random Pauli graphs}\label{sec:graph}
In this section, we will prove the following result on the expected edge weight of the random Pauli graph.
\begin{theorem}[Edge Weight of Random Pauli Graph]\label{thm:pauli-density} Let $\lambda = \frac{1+\sqrt{33}}{8}$, then
\begin{align*}
    \max_{\target \in \statespace} \randgraph{\target}{\gamma} \leq \left(\frac{4}{3}\right)^\nqubits \cdot 4^{\nqubits\left(\frac{\lambda\gamma-1}{(1+\lambda)(1+\gamma)}\right)}
\end{align*} \\
for $\gamma \in [0,\lambda]$.
\end{theorem}
\begin{proof}
    We will utilize the purity constraints of $\target$ to bound the expected edge weight. We first recognize that the graph has a tensor structure since $Q \cap Q'$ is a qubit-wise operation.
    \begin{align*}
        \randgraph{\target}{\gamma} &= \overlap{\nu_\gamma}{A^{\otimes \nqubits} \nu_\gamma},
    \end{align*}
    where $\nu_\gamma$ is the vectorization $\left\{|\ftarget{Q}|^{\frac{2}{1+\gamma}}\right\}_{Q \in \pauliset}$, and $A$ is the weighted adjacency matrix:
    \renewcommand{\arraystretch}{1.25}
    \begin{align*}
        A &= \begin{bmatrix}
            1 & 0 & 0 & \frac{1}{3} \\
            0 & 1 & 0 & \frac{1}{3} \\
            0 & 0 & 1 & \frac{1}{3} \\
            \frac{1}{3} & \frac{1}{3} & \frac{1}{3} & \frac{1}{3} 
        \end{bmatrix} 
            = \frac{4}{3} \cdot \begin{bmatrix}
            \frac{3}{4} & 0 & 0 & \frac{1}{4} \\
            0 & \frac{3}{4} & 0 & \frac{1}{4} \\
            0 & 0 & \frac{3}{4} & \frac{1}{4} \\
            \frac{1}{4} & \frac{1}{4} & \frac{1}{4} & \frac{1}{4} 
        \end{bmatrix}.
    \end{align*}
    Here, we designate the row and column orders to be $X,Y,Z,I$. Maximizing the edge weight corresponds to having more intersecting non-$I$ terms. With factorization of $\frac{4}{3}$, we see that $A$ can be interpreted as a transition probability matrix associated with a noisy channel. For entries $\neq I$, the channel behaves as an erasure channel with the $I$ character. When the entry is $I$, the channel becomes completely uniform. We propose to bound the channel matrix by the binary symmetric channel to establish hypercontractivity.
    \begin{lemma}[Symmetrizing Noisy Channel]~\label{lem:noise-sym}Let $p \leq \frac{7-\sqrt{33}}{16}$, then
    \begin{align*}
    BSC_{p}^{\otimes 2} \succeq 
         \begin{bmatrix}
           \frac{3}{4} & 0 & 0 & \frac{1}{4}  \\
           0 & \frac{3}{4} & 0 & \frac{1}{4} \\
           0 & 0 & \frac{3}{4} & \frac{1}{4} \\
           \frac{1}{4} & \frac{1}{4} & \frac{1}{4} & \frac{1}{4} \\
        \end{bmatrix},
    \end{align*}
    where $BSC_p = \begin{bmatrix}
        1-p & p \\
        p & 1-p
    \end{bmatrix}$.
    \end{lemma}
    The main idea behind the proof of this lemma is solving for the eigenvalues of $BSC_p^{\otimes 2} - \frac{3}{4} A \succeq 0$ and finding the conditions of $p$ that result in the eigenvalues to be non-negative. The proof will be deferred to~\cref{pf:noise-sym}. Fixing $p^* = \frac{7-\sqrt{33}}{16}$, we can bound the random graph quantity,
    \begin{align*}
          \randgraph{\target}{\gamma} &\leq \left(\frac{4}{3}\right)^{\nqubits} \overlap{\nu_\gamma}{BSC_{p^*}^{\otimes 2\nqubits} \nu_\gamma},
    \end{align*}
    since $A \preceq B \implies A^{\otimes k} \preceq B^{\otimes k}$ when $A,B \succeq 0$. In the context of Boolean analysis, the product binary symmetric channel can be interpreted as the Bonami-Beckner operator over $2 \nqubits$ bits, which is defined as $T_{1-2p} = BSC_{p}^{\otimes 2 \nqubits}$. Thus, our inner product can alternatively be written as $\stab{T_\gamma}{\nu_\gamma} = \overlap{\nu_\gamma}{T_{\lambda} \nu_\gamma}$, where $\lambda = 1-2p^*  = \frac{1+\sqrt{33}}{8}$.
    
    From~\cref{sec:boolean},  we can see that the operator forms a Markov semi-group, which means that $T_{\lambda_1}T_{\lambda_2} = T_{\lambda_1 \lambda_2}$. Combined with the fact that the operator is symmetric,
    \begin{align*}
         \randgraph{\target}{\gamma} &\leq \left(\frac{4}{3}\right)^{\nqubits} \stab{T_{\sqrt{\lambda}}}{\nu_\gamma}=\left(\frac{4}{3}\right)^{\nqubits} \|T_{\sqrt{\lambda}} \nu_\gamma\|_2^2.
    \end{align*}
    Now, we can apply hypercontractivity. By~\cref{lem:hypercontractivity}, the 2-norm under the noisy channel contracts the $(1+\lambda)$-norm of $\nu_\gamma$.
    \begin{align*}
        \randgraph{\target}{\gamma} &\leq \left(\frac{4}{3}\right)^{\nqubits}  \cdot 4^{\nqubits\left(1 - \frac{2}{1+\lambda}\right)} \cdot \|\nu_{\gamma}\|_{1+\lambda}^2 \\
        &\leq  \left(\frac{4}{3}\right)^{\nqubits} \cdot 4^{N\left(\frac{\lambda - 1}{1+\lambda}\right)} \cdot \|\nu_{\gamma}\|_{\infty}^{2 \left(\frac{\lambda - \gamma}{1+\lambda}\right)} \cdot \|\nu_{\lambda}\|_{1+\lambda}^2 \\
        &= \left(\frac{4}{3}\right)^{\nqubits} \cdot 4^{N\left(\frac{\lambda-1}{1+\lambda}\right)} \cdot \|\hat{\target}\|_{\infty}^{4 \left(\frac{\lambda-\gamma}{(1+\lambda)(1+\gamma)}\right)} \cdot \|\hat{\target}\|_2^{\frac{4}{1+\lambda}}
    \end{align*}
    We can now use the p.s.d and purity constraints to establish the final bound of the expected edge weight of the Pauli graph. By~\cref{obs:pauli-constraints},
    \begin{align*}
        \randgraph{\target}{\gamma} &\leq  \left(\frac{4}{3}\right)^{\nqubits} \cdot 4^{\nqubits\left(\frac{\lambda -1}{1+\lambda} + \frac{\gamma - \lambda}{(1+\lambda)(1+\gamma)}\right)} \\
        &=  \left(\frac{4}{3}\right)^{\nqubits}4^{\nqubits\left(\frac{\gamma \lambda - 1}{(1+\lambda)(1+\gamma)}\right)}.
    \end{align*}
    We have proven~\cref{thm:pauli-density}.
\end{proof}

\subsection{Bounding lab state norm}\label{sec:lab-norm}
We will now bound the lab state norm. 

\begin{lemma}[Lab State Norm]\label{lem:lab-norm} For any $0 \leq \gamma < 1$,
\begin{align*}
    \max_{\lab \in \statespace} \knap{\lab} \leq  2^{N\left(\frac{1-2\gamma}{1-\gamma}\right)}.
\end{align*}
\end{lemma}

\noindent We first simplify the norm with counting arguments. Then, we establish the final bound by recognizing that the expression is related to the purity of the lab state under the single-qubit depolarizing channel.
\begin{proof}
We will first simplify $\knap{\lab}$,
\begin{align*}
    \knap{\lab} = \sum_{(Q,Q') \in \qcommuteset} \frac{|\hat{\lab}(Q \Delta Q')|^{\frac{2}{1-\gamma}}}{3^{\nqubits - w(Q \cap Q')}} = \frac{1}{3^\nqubits} \sum_{Z} |\hat{\lab}(Z)|^{\frac{2}{1-\gamma}} \sum_{Q \Delta Q' = Z} 3^{w(Q \cap Q')}.
\end{align*}
For a fixed $Z$, $Q,Q'$ must be disjoint along the non-identity positions of $Z$ in order for $Q \Delta Q' = Z$, in which there are $2^{w(Z)}$ ways to split the non-identity entries of $Z$. For the identity positions of $Z$, $Q$ and $Q'$ must be identical to be canceled out under the symmetric difference. Therefore,
\begin{align}
    \knap{\lab} &= \frac{1}{3^\nqubits} \sum_{Z \in \mathcal{Q}} |\hat{\lab}(Z)|^{\frac{2}{1-\gamma}} \cdot 2^{w(Z)} \sum_{w=0}^{\nqubits-w(Z)} 3^w \cdot 3^w \cdot \binom{N-w(Z)}{w} \\
    &= \left(\frac{10}{3}\right)^\nqubits \sum_{Z \in \mathcal{Q}} \frac{|\hat{\lab}(Z)|^{\frac{2}{1-\gamma}}}{5^{w(Z)}} \leq \left(\frac{10}{3}\right)^\nqubits \cdot \|\hat{\lab}\|_\infty^{\frac{2 \gamma}{1-\gamma}} \sum_{Z \in \pauliset} \frac{\hat{\lab}(Z)^{2}}{5^{w(Z)}} \\
    &\leq \left(\frac{10}{3}\right)^{\nqubits} \cdot 2^{-\nqubits\left(\frac{\gamma}{1-\gamma}\right)} \sum_{Z \in \pauliset} \frac{\hat{\lab}(Z)^{2}}{5^{w(Z)}} \\
    &=  \left(\frac{10}{3}\right)^{\nqubits} \cdot 2^{-\nqubits\left(\frac{\gamma}{1-\gamma}\right)}\left\|D_{\sqrt{{5}}^{-1}}^{\otimes \nqubits}(\lab)\right\|_2^2 \label{eq:lab-norm},
\end{align}
where $D_s(W) = s \cdot W + (1-s) \Tr[W] \cdot \frac{I}{2}$ is the single-qubit depolarizing channel. We will use a powerful tensorization property for the maximum $\ell_p$ norms under $D_{s}(\cdot)$~\cite[Theorem 3]{king2003capacity}, which was the main ingredient to establish additivity of Holevo information under single-qubit depolarization channels.

\begin{lemma}[{\cite[Theorem 3]{king2003capacity}}] Given any quantum channel $\Psi: \mathcal{D}(\mathbb{C}^{d/2}) \rightarrow \mathcal{D}(\mathbb{C}^{d/2})$,
\begin{align*}
v_p(D_s \otimes \Psi) \leq v_p(D_s) \cdot v_p(\Psi),
\end{align*}
where $v_p(A) = \max_{\target \in \statespace} \|A(\target)\|_p$ for $p \geq 1$.
\end{lemma}
\noindent By inducting on $\nqubits$ for $\Psi = D_s^{\otimes \nqubits-1}$, we can reduce our purity bound to the maximum purity for single-qubit states.
\begin{align*}
    \left\|D_{\sqrt{{5}}^{-1}}^{\otimes \nqubits}(\lab)\right\|_2 \leq \left(\max_{\target \in \mathcal{D}(\mathbb{C}^2)} \|D_{\sqrt{{5}}^{-1}} (\target)\|_2\right)^{\nqubits}.
\end{align*}
The single-qubit purity will be the following,
\begin{align*}
    \|D_{\sqrt{5}^{-1}}(\target)\|_2 &= \sqrt{\frac{1}{2} + \frac{1}{5} (\ftarget{X}^2 + \ftarget{Y}^2 + \ftarget{Z}^2)} \leq \sqrt{\frac{3}{5}},
\end{align*}
where we used the fact that $\ftarget{X}^2 + \ftarget{Y}^2 + \ftarget{Z}^2 \leq \frac{1}{2}$. Thus, $\left\|D_{\sqrt{{5}}^{-1}}^{\otimes \nqubits}(\lab)\right\|_2^2 \leq \left(\frac{3}{5}\right)^{\nqubits}$, which holds tightly when $\lab$ is a Pauli projector state. The lemma is proven when we apply the bound to~\cref{eq:lab-norm}.

\end{proof}
\subsection{Putting it all together}
We now conclude the sublinear bound on the biased Pauli norm and prove~\cref{thm:bias-norm}. 
\begin{proof}[Proof of~\cref{thm:bias-norm}]
We will fix $\gamma=\lambda$ for~\cref{lem:reduction} and apply~\cref{thm:pauli-density} and~\cref{lem:lab-norm}. The result is
\begin{align*}
    \biasnorm{\lab}{\target} &\leq \sqrt{\left(\frac{4}{3}\right)^{(1+\lambda)\nqubits} \cdot 4^{(\lambda - 1) \nqubits} \cdot 2^{-2 \lambda\nqubits}} \cdot \|\target\|_2^2 \\
    &= \frac{2^{\lambda \nqubits} \cdot \sqrt{3}^{(1-\lambda) \nqubits}}{3^\nqubits} 
\end{align*}
Thus, we have proven~\cref{thm:bias-norm}.
\end{proof}

%% file: dipe.tex
\section{Sublinear Distributed Inner Product Estimation}~\label{sec:dipe}
It turns out that we can use the same techniques in~\cref{sec:biased-norm} to establish sublinear guarantees for Distributed Inner Product Estimation. The idea is to improve the analysis of the single-qubit shadow estimator used in~\cite{huang2026local}.
\begin{extraprotocol}{P2: Distributed Inner Product Estimation with Pauli Shadows~\cite{huang2026local}}{}
 \textbf{Input: } Measurement access to $\lab^{\otimes n}$ and $\target^{\otimes n}$. 
 \bigskip
 
 \textbf{Basis sampling. } Sample $B_1, B_2, \ldots, B_n$ and $B'_1, B'_2, \ldots, B'_n$ uniformly from $\basesset$.
\bigskip

\textbf{Measurement postprocessing. } For each $i \in [n]$, measure copy $\lab_i$ and $\target_i$ with $\POVM_{B_i}$ and obtain the outcomes $X_{B_i} \in \{-1,1\}^\nqubits$, respectively. Then, construct the state shadows
\medskip
$$\Tilde{\lab}_i = \frac{1}{d} \sum_{S \subseteq [\nqubits]} 3^{|S|}\parity{S}{X_{B_i}} B_i^S.$$
Repeat the same process for $\target_i$ under bases $B'_1, B'_2, \ldots, B'_n$. \\

\textbf{Combining the outcomes. } Return $\Tilde{F}(\target, \lab) = \frac{1}{n^2}\overlap{\sum_{i=1}^n \Tilde{\lab}_i}{\sum_{i=1}^n \Tilde{\target}_i}$.
\end{extraprotocol}
 We will show that the variance of the protocol can be improved to attain sublinear guarantees in the high precision regime.
\addtocounter{cor}{-1}
\begin{cor}[Sublinear Distributed Inner Product Estimation with Pauli Basis Measurements] Protocol P2 estimates the inner product $\overlap{\target}{\lab}$ with 
\begin{align*}
n = \bigO{\left(\frac{\sqrt{7.5}^\nqubits}{\eps} \lor \frac{2^{\lambda \nqubits} \sqrt{3}^{(1-\lambda) \nqubits}}{\eps^2}\right) \log \frac{1}{\delta}}
\end{align*}
copies of $\target$ and $\lab$, where $\lambda  = \frac{1 + \sqrt{33}}{8}$.
\end{cor}
\begin{proof}
    We continue with the proof presented in~\cite[Appendix E]{huang2026local}. They show the protocol produces an unbiased estimate of the inner product and bound the variance by the following quantity.
    \begin{align*}
        \Var{\Tilde{F}(\target, \lab)} &\leq \frac{\expectDistrOf{}{\Tr[\Tilde{\target} \Tilde{\lab}]^2}}{n^2} + \frac{\expectDistrOf{}{\Tr[\target \Tilde{\lab}]^2} + \expectDistrOf{}{\Tr[\Tilde{\target} \lab]^2}}{n} \\
        &\leq \frac{7.5^\nqubits}{n^2} + \frac{\left(\shadownorm{\target}\right)^2 + \left(\shadownorm{\lab}\right)^2}{n} 
    \end{align*}
     By~\cref{thm:bias-norm}, we have
    \begin{align*}
        \Var{\Tilde{F}(\target, \lab)} \leq \frac{7.5^\nqubits}{n^2} + 2 \cdot \frac{2^{\lambda \nqubits} \cdot \sqrt{3}^{(1-\lambda)\nqubits}}{n}
    \end{align*}
    Setting $n = 8\left(\frac{\sqrt{7.5}^\nqubits}{\eps} \lor \frac{2^{\lambda \nqubits} \sqrt{3}^{(1-\lambda)}}{\eps^2}\right)$ and using Medians-of-Means proves the copy-complexity for Distributed Inner Product Estimation.
    
\end{proof}

%% file: additional.tex
\section{Proof of Noisy Channel Symmetrization}\label{pf:noise-sym}
We now provide the proof for the noise channel symmetrization, which allows us to establish our noise-stability bounds with Boolean hypercontractivity.
\addtocounter{lemma}{-3}
\begin{lemma}[Symmetrizing Noisy Channel] Let $p \leq \frac{7-\sqrt{33}}{16}$, then
    \begin{align*}
    BSC_{p}^{\otimes 2} \succeq 
         \begin{bmatrix}
           \frac{3}{4} & 0 & 0 & \frac{1}{4}  \\
           0 & \frac{3}{4} & 0 & \frac{1}{4} \\
           0 & 0 & \frac{3}{4} & \frac{1}{4} \\
           \frac{1}{4} & \frac{1}{4} & \frac{1}{4} & \frac{1}{4} \\
        \end{bmatrix},
    \end{align*}
    where $BSC_p = \begin{bmatrix}
        1-p & p \\
        p & 1-p
    \end{bmatrix}$.
\end{lemma}
\addtocounter{lemma}{1}
\begin{proof} We will define the following,
\begin{align*}
 C = \begin{bmatrix}
           \frac{3}{4} & 0 & 0 & \frac{1}{4}  \\
           0 & \frac{3}{4} & 0 & \frac{1}{4} \\
           0 & 0 & \frac{3}{4} & \frac{1}{4} \\
           \frac{1}{4} & \frac{1}{4} & \frac{1}{4} & \frac{1}{4}
        \end{bmatrix}.
\end{align*}
The strategy will be to find the largest $p$ such that $BSC_p^{\otimes 2} \succeq C$ for $p \in [0, \frac{1}{2})$ . To simplify this optimization, we will first consider the unitary diagonalization of $BSC_p$,
\begin{align*}
    BSC_p = H \cdot diag(1, 1-2p) \cdot H,
\end{align*}
where $H$ is the Hadamard matrix. Thus, the unitary diagonalization of $BSC_p^{\otimes 2}$ is $H^{\otimes 2} \cdot diag(1, 1-2p)^{\otimes 2} \cdot H^{\otimes 2}$. We will apply this unitary transformation to $C$,
\begin{align*}
    H^{\otimes 2} C H^{\otimes 2} = \begin{bmatrix}
        1 & 0 & 0 & 0 \\
        0 & \frac{1}{2} & -\frac{1}{4} & \frac{1}{4} \\
        0 & -\frac{1}{4} & \frac{1}{2} & \frac{1}{4} \\
        0 & \frac{1}{4} & \frac{1}{4} & \frac{1}{2}
    \end{bmatrix}
\end{align*}
It suffices to show that $H^{\otimes 2} C H^{\otimes 2} \preceq diag(1, 1-2p)^{\otimes 2}$ to show that $ C \preceq BSC_p^{\otimes 2}$ since the condition is equivalent up to rotation. Looking at the difference matrix $\Delta_p \eqdef diag(1, 1-2p)^{\otimes 2}-H^{\otimes 2} C H^{\otimes 2}$,
\begin{align*}
    \Delta_p = \begin{bmatrix}
        0 & 0 & 0 & 0 \\
        0 & \frac{1}{2}-2p & \frac{1}{4} & -\frac{1}{4} \\
        0 & \frac{1}{4} & \frac{1}{2}-2p & -\frac{1}{4} \\
        0 & -\frac{1}{4} & -\frac{1}{4} & \frac{1}{2} - 4p + 4p^2
    \end{bmatrix} \eqdef \begin{bmatrix}
        0 & 0 \\
        0 & S_p
    \end{bmatrix}
\end{align*}
To check $4\Delta_p \succeq 0$, it is only necessary to check the positive semi-definiteness of 
\begin{align*}
4 S_p = \begin{bmatrix}
    2-8p & 1 & -1 \\
    1 & 2-8p & -1 \\
    -1 & -1 & 2-16p+16p^2
\end{bmatrix}.
\end{align*}

\noindent We can see that one of the eigenvectors is $\begin{pmatrix} 1 \\ -1 \\ 0\end{pmatrix}$ with eigenvalue $1-8p$, imposing the condition that $p \leq \frac{1}{8}$. Therefore, the other eigenvectors will be in the orthogonal complement: $span\{w_1, w_2\}$. We choose the orthonormal basis $w_1, w_2$ to be $\begin{pmatrix}
    0 \\ 0 \\ 1
\end{pmatrix}, \frac{1}{\sqrt{2}}\begin{pmatrix}
    1 \\ 1 \\ 0
\end{pmatrix}$. We will now be solving the following eigenvalue problem,
\begin{align*}
    \lambda (c_1 w_1 + c_2 w_2) &= 4S_p (c_1 w_1 + c_2 w_2) \\
    &\boldsymbol{\Longleftrightarrow} \\
    \lambda \begin{bmatrix}
        w_1 & w_2
    \end{bmatrix} \begin{pmatrix}
        c_1 \\ c_2
    \end{pmatrix} &= 4S_p \begin{bmatrix}
        w_1 & w_2
    \end{bmatrix} \begin{pmatrix}
        c_1 \\ c_2
    \end{pmatrix} \\
    &\boldsymbol{\Longleftrightarrow} \\
    \lambda \begin{pmatrix}
        c_1 \\ c_2
    \end{pmatrix} &=  \begin{bmatrix}
        w_1 & w_2
    \end{bmatrix}^T 4S_p \begin{bmatrix}
        w_1 & w_2
    \end{bmatrix} \begin{pmatrix}
        c_1 \\ c_2
    \end{pmatrix}.
\end{align*}
So, it suffices to show that $\begin{bmatrix}
        w_1 & w_2
    \end{bmatrix}^T 4S_p \begin{bmatrix}
        w_1 & w_2
    \end{bmatrix} \succeq 0$. Applying the isometry to $4S_p$ yields
    \begin{align*}
        \begin{bmatrix}
        w_1 & w_2
    \end{bmatrix}^T 4S_p \begin{bmatrix}
        w_1 & w_2
    \end{bmatrix} = \begin{bmatrix}
        2 - 16p + 16p^2 & -\sqrt{2} \\
        -\sqrt{2} & 3-8p
    \end{bmatrix}.
    \end{align*}
     By generalized Sylvester's criterion, we can check if the determinants of every principal minor is non-negative to ensure positive semi-definiteness. The diagonals (or the $1 \times 1$ principal minors) must be non-negative, so $3-8p \geq 0$ and $2 - 16p + 16p^2 \geq 0$. In addition, the determinant of the whole matrix (the only $2 \times 2$ principal minor) must also be non-negative.
    \begin{align*}
        (3-8p)(2-16p&+16p^2)-2 \geq 0 \\
        &\boldsymbol{\Longleftrightarrow} \\
        (3-8p)(1-8p+8p^2) &- 1 = 2 - 32p + 88p^2 - 64p^3\geq 0
    \end{align*}
    We can factor out the cubic to $2(1-2p)(16p^2 -14p + 1)$. Since $p \in [0, \frac{1}{2})$, $(1-2p) > 0$ can be divided on both sides for the resulting condition, $(16p^2 -14p + 1) \geq 0$. There are $4$ conditions that must be satisfied in order for $BSC_p^{\otimes 2} \succeq C$.
    \begin{align*}
        p &\leq 1/8 \\
        p &\leq 3/8 \\
        2 - 16p &+ 16p^2 \geq 0 \\
        1 - 14p &+ 16p^2 \geq 0
    \end{align*}
    The second condition is vacuous given the first condition.  It suffices to check if $2 - 16p + 16p^2 \geq 0$ and $1 - 14p + 16p^2 \geq 0$. Solving for the roots of the second quadratic, we get
    \begin{align*}
        p = \frac{7 \pm \sqrt{33}}{16}.
    \end{align*}
    In the interval $[0,\frac{1}{2})$, it suffices to have $p \leq  \frac{7 - \sqrt{33}}{16}$ to ensure the quadratic is non-negative. For the first quadratic,
    \begin{align*}
        p = \frac{2 \pm \sqrt{2}}{4}.
    \end{align*}
    are the roots. Similarly, $p \leq \frac{2 - \sqrt{2}}{4}$ suffices to guarantee quadratic is non-negative. Having $p \leq  \frac{7 - \sqrt{33}}{16} < 1/8 < \frac{2 - \sqrt{2}}{4}$ ensures all of the conditions hold. Thus, the proof is complete.
\end{proof}